\documentclass[letterpaper, 10 pt, conference]{ieeeconf}
\usepackage[protrusion=true,expansion=true]{microtype} 
\IEEEoverridecommandlockouts                              

\usepackage[pdftex]{graphicx}
\DeclareGraphicsExtensions{.pdf}
\graphicspath{{./media/}}
\usepackage{xcolor}

\usepackage[cmex10]{amsmath}
\usepackage{amsthm}
\usepackage{amssymb}
\usepackage[ruled,vlined]{algorithm2e}
\usepackage{balance}

\usepackage{array}
\usepackage{mdwmath}
\usepackage{mdwtab}
\usepackage{eqparbox}
\usepackage{blkarray}

\usepackage[british]{babel}

\usepackage{siunitx}

\usepackage{cleveref}

\makeatletter
\let\MYcaption\@makecaption
\makeatother

\usepackage[font=footnotesize]{subcaption}

\makeatletter
\let\@makecaption\MYcaption
\makeatother

\usepackage{dblfloatfix}

\usepackage{cite}

\usepackage{url}

\newtheorem{theorem}{Theorem}

\newtheorem{lemma}{Lemma}

\theoremstyle{definition}
\newtheorem{assumption}{Assumption}
\newtheorem{definition}{Definition}

\newtheorem{remark}{Remark}

\usepackage{amssymb, amsmath, latexsym, amsfonts, mathrsfs, amsbsy}
\usepackage{color}
\usepackage{bm}

\newcommand{\eqb}[1]{\begin{equation}\label{#1}}
\newcommand{\eqe}{\end{equation}}

\newcommand{\mbb}[1]{\mathbb{#1}}
\newcommand{\mc}[1]{\mathcal{#1}}

\newcommand{\ib}{\begin{itemize}}
\newcommand{\ie}{\end{itemize}}

\def \symm {+}
\newsavebox{\ieeealgbox}

\title{\LARGE \bf
Distributed Adaptive and Resilient Control of Multi-Robot Systems \\with Limited Field of View Interactions using Q-Learning
}

\author{Pratik Mukherjee$^{1}$, Matteo Santilli$^{2}$, Andrea Gasparri$^{2}$, and Ryan K. Williams$^{1}$%
\thanks{$^{1}$P. Mukherjee and R.K. Williams are with Electrical and Computer Engineering Department, Virginia Polytechnic Institute and State University, Blacksburg, VA USA, {\tt\small \{mukhe027, rywilli1\}@vt.edu} }%
\thanks{$^{2}$M. Santilli and A. Gasparri are with the Engineering Department, Roma Tre University, Roma, 00146, Italy, {\tt\small matteo.santilli@uniroma3.it, gasparri@dia.uniroma3.it}}%
}

\def\trifov#1{\mathcal{T}_{#1}}

\def\dims{3}
\def\dimp{2}

\begin{document}

\maketitle
\thispagestyle{empty}
\pagestyle{empty}

\begin{abstract}
 In this paper, we consider the problem of dynamically tuning gains for multi-robot systems (MRS) under potential based control design framework where the MRS team coordinates to maintain a connected topology while equipped with limited field of view sensors.  Applying the potential-based control framework and assuming robot interaction is encoded by a triangular geometry, we derive a distributed control law in order to achieve the topology control objective. A typical shortcoming of potential-based
control in distributed networks is that the overall system behavior is highly sensitive to gain-tuning.
To overcome this limitation, we propose a distributed and
adaptive gain controller that preserves a designed pairwise interaction strength, independent of the network size. Over that, we implement a control scheme that enables the MRS to be resilient against exogenous attacks on on-board sensors or actuator of the robots in MRS. In this regard, we model additive sensor and actuator faults  which are induced externally to render the MRS unstable. However, applying  $H_{\infty}$ control protocols by employing a static output-feedback design technique guarantees bounded $L_2$ gains of the error induced by the sensor and actuator fault signals. Finally, we apply policy iteration based Q-Learning to solve for adaptive gains for the discrete-time MRS. Simulation results are provided to support the theoretical findings.

\end{abstract}

\section{Introduction}
Recent efforts in distributed multi-robot coordination have
exploded, both in number and in breadth of capability. The
significant interest in distributed systems is not surprising as
the applications are numerous, including for example sensor
networks, collaborative robotics, transportation, etc. One trend
of multi-robot control has been potential-based design. Indeed, the potential-based control
methodology has been widely investigated by the control community, e.g.,
\cite{Gasparri:CDC:2017:1,mukherjee2019experimental,santilli2019distributed,ji2007distributed,zavlanos2007potential,dimarogonas2008connectedness} . The advantages of potential-based design
is that controllers are easily distributed, physically-motivated,
and come with provable convergence guarantees.
Although the potential-based control framework is very powerful, it exhibits several shortcomings ranging from local minima to high sensitivity to gains. While local minima are inherent to the design of potentials encoding network objectives,
high sensitivity to gains can be mitigated with a proper choice
of pairwise interaction strength for a given network. 
To the best of our knowledge, the closest related work is represented by recent results on performance-driven multi-robot systems (MRS). Briefly, the underlying idea of these approaches is
to enforce the transient response of a system to match some
designer-imposed metric. Works that approach this problem
have generally selected performance metrics in two distinct
manners: either by restricting upper and lower bounds on overall convergence rate to equilibrium, or by constraining instantaneous control action, such as maximum relative velocity.
Works such as \cite{bechlioulis2014low}, \cite{bechlioulis2014robust} (which are based on the prescribed
performance framework) and \cite{karayiannidis2012multi} take the former approach, controlling the transient response of the error signal
in formation control and consensus, respectively. For example,
\cite{bechlioulis2014robust} achieves robust model-free prescribed control for leader-follower formations, which is decoupled from graph topology,
control gain selection, and model uncertainties. Alternatively,
works such as \cite{korda2014stochastic}, \cite{guo2013controlling} take a model-driven approach in controlling constraint violation, \cite{korda2014stochastic} applying model-predictive
control and \cite{guo2013controlling} through navigation functions (a form of potential control), both in a centralized context. 

One other area in directed coordination problems that has received recent attention \cite{Asadi:2016,Zhang:2013} is control with limited \emph{fields of view} (FOVs), for example aerial vehicles equipped with cameras capable of sensing other robots in a visually restricted area. 
In our recent work \cite{santilli2019distributed} we  exhibit a distributed potential-based coordination framework for an MRS equipped with limited field of view sensors to maintain a asymmetric, connected topology. In this context, threat from external attacks that can induce sensor and actuator faults to destabilize an already sensitive to gains potential based control system becomes real. Limited FOV control applications of coordinated MRS of the type developed in \cite{santilli2019distributed} are not just prone to the sensitivity of the potential based controllers but also under constant threat of external attacks especially in case of asymmetric interaction (i.e. directed graphs) control where stability is not inherent as proposed in our prior works \cite{Gasparri:CDC:2017:1,mukherjee2020optimal}. Therefore, it is important that we develop a control scheme that is resilient to such external attacks. Works in literature such as \cite{chen2020adaptive,chen2019resilient} address the problem of solving for resilient controllers of multi-agent systems that are exposed to external attacks, faults or disturbances. Particularly, we will now briefly review the work in \cite{chen2019resilient} where the authors address the problem of resilient and adaptive multi-agent systems under actuator and sensor faults. The nature of the faults are assumed to be additive and particularly the sensor fault is assumed to be unbounded. The authors use the \emph{static output feedback} technique to show that the error induced from the faults has bounded $L_2$ gains in terms of the $L_2$ norms of fault signals. In this paper, we modify the static output feedback technique using $H_{\infty}$ control protocols to fit our potential based control framework as in \cite{santilli2019distributed}
to tackle additive sensor and actuator fault that can be catastrophic for a sensitive MRS with asymmetric interaction topology where topology stability is not inherent unlike in the case of symmetric (i.e. undirected graphs) interactions as also detailed in \cite{Gasparri:CDC:2017:1}. 

In this direction, our contribution is fourfold.
i) We demonstrate that our distributed controller yields
an approximation of the desired pairwise interaction strength,
while preserving the equilibria of an underlying network objective for robots coordinating in asymmetric interaction scenarios with limited FOVs. 
 ii) We couple the potential based controller with $H_{\infty}$ control protocols, similar to the one designed in \cite{chen2019resilient} to tackle external attacks on the MRS that can induce faults in sensors and actuators of the robots, such that we can show that the induced error  from the faults has bounded $L_2$ gains in terms of the $L_2$ norms of fault signals and the underlying objective can remain intact. iii) We demonstrate the implementation of \emph{policy iteration} based Q-Learning method for computing the pairwise gains $k_{ij}$ online and prove convergence to the optimal gains. iv) We corroborate the theoretical results with extensive simulations.


\section{Preliminaries}\label{sec:preliminaries}

\subsection{Robot Network Modeling} \label{sub_sec:mod}

Let us consider an MRS composed of $n$~robots and assume that each robot~$i$ has the a first-order dynamics $\dot{s}_i(t) = u_i(t)$
%
%
with $s_i(t) = [p_i(t)^T, \, \theta_i(t)]^T \in \mathbb{R}^{\dimp} \times (-\pi,\pi]$ the state of the robot~$i$ composed of the position $p_i(t) = [x_i(t), y_i(t)]^T \in \mathbb{R}^{\dimp}$ and the orientation $\theta_i \in (-\pi,\pi]$, while $u_i(t) \in \mathbb{R}^{\dims}$ denotes the control input. Stacking robot states and inputs yields the overall system $\dot{\mathbf{s}}(t) = \mathbf{u}(t)$
%
%
with $\mathbf{s}(t) = [s_1(t)^T,\, \ldots, \, s_n(t)^T]^T \in \mathbb{R}^{\dimp n} \times (-\pi, \pi]^n$ and $\mathbf{u}(t) = [u_1(t)^T,\, \ldots, \, u_n(t)^T]^T \in \mathbb{R}^{\dimp n} \times (-\pi, \pi]^n$ the stacked vector of states and control inputs, respectively. In the sequel, time-dependence will be omitted for the sake of clarity. 


Let us assume that each robot~$i$ possesses a limited field of view that is encoded by a triangle geometry $\trifov{i}$ rigidly fixed to the robot as also modelled in \cite{santilli2019distributed}. This kind of sensing yields asymmetric robot interactions that we will describe through a directed graph \mbox{$\mathcal{G} = \{ \mathcal{V}, \, \mathcal{E} \}$} with node set $ \mathcal{V} = \{ q_1,\ldots,q_n \}$ and edge set $\mathcal{E} \subseteq \mathcal{V} \times \mathcal{V}$. In particular, we will say that an edge $e_{ij} \in \mathcal{E}$ connects robot~$i$ and robot~$j$ if $p_j \in \trifov{i}$. In addition, when referencing single edges we will use the convention $e_k$ meaning that we are referencing the $k$-th directed edge out of $|\mathcal{E}|$ total edges\footnote{In order to reference the $k$-th edge, the edge set $\mathcal{E}$ needs to be sorted. A simple sort can be obtained enumerating the edges $(1,j)$ of the first robot as \mbox{$e_1,e_2,\ldots,$} then the edges $(2,h)$ of the second robot and so on.}. Moreover, we will denote by \mbox{$\mathcal{N}_{i}^{+} = \{j \in \mathcal{V} :(i,j) \in \mathcal{E}\}$} the set of \emph{out-neighbors} of robot~$i$ and \mbox{$\mathcal{N}_{i}^{-} = \{j \in \mathcal{V} :(j,i) \in \mathcal{E}\}$} the set of \emph{in-neighbors}. Note that since the graph is directed $(i,j) \in \mathcal{E}$ does not imply $(j,i) \in \mathcal{E}$. 

A useful representation for a directed graph $\mathcal{G}$ is the \textit{incidence matrix} $ \mathcal{B}(\mathcal{G})\in \mathbb{R}^{n\times |\mathcal{E}|}$, that is a matrix with rows indexed by robots and columns indexed by edges, such that $\mathcal{B}_{ij} = 1$ if the edge $e_j$ leaves vertex $v_i$, $-1$ if it enters vertex $v_i$, and $0$ otherwise. The \emph{outgoing incidence matrix} $\mathcal{B}_+$ contains only the outgoing parts of the incidence matrix $\mathcal{B}$, with incoming parts set to zero. We will also make use of the \emph{directed edge Laplacian} $\mathcal{L}_{\mathcal{E}}^d \in \mbb{R}^{|\mc{E}| \times |\mc{E}|}$ given by $\mc{L}_{\mc{E}}^d = \mc{B}^T\mc{B}_+$.  For properties of the edge Laplacian see for example \cite{Zelazo:2007,Zeng:2016}.

\section{Directed Coordination Framework}\label{sec:control}
In this section , we briefly review important results from our previous work \cite{santilli2019distributed} which are useful in understanding the theoretical background of this paper. We refer the reader to \cite{santilli2019distributed} for the detailed theory of potential based limited FOV control of MRS.

\subsection{Potential Fields for Topology Maintenance}

As a case study let us consider the maintenance of the interactions (topology) among the robots in the network, i.e., we want to preserve the initial graph~$\mathcal{G}$.  In other words, we want to derive a distributed potential-based control law such that each robot~$i$ maintains its neighbor~$j$ inside the triangle~$\trifov{i}$ that defines the perimeter of the FOV of each robot as detailed in \cite{santilli2019distributed}. In this direction, we can define the potential field term $\Phi_{ij}(s_i,s_j)$ which encodes the sum of energies of the perpendicular distances from a point, the agent~$j$, to the line, the sides of the triangle~$\trifov{i}$.
Next, we model a quality of interaction map  by using a two-dimensional Gaussian function $\Psi_{ij}(s_i,s_j)$ with mean  and variance.
For the FOV maintenance controller, we can move along the anti-gradient of our potential fields $\Phi_{ij}(s_i,s_j)$ and $\Psi_{ij}(s_i,s_j)$ to orient the FOV of the robots in order to minimize the distance between the neighbors $j \in \mathcal{N}_i$ and the desired position. We now provide two technical results combined of the potential fields $\Phi_{ij}$ and $\Psi_{ij}$ that will prove  necessary to derive the main result in Theorem~\ref{th:1}. 

\begin{lemma}(Lemma 1 and 2 in \cite{santilli2019distributed})\label{lemma:potential-property-psi}
The potential fields $\Phi_{ij}$ and $\Psi_{ij}$ satisfy the local potential properties $\nabla_{p_i}\Phi_{ij} = -\nabla_{p_j}\Phi_{ij}$ and $\nabla_{p_i}\Psi_{ij} = -\nabla_{p_j}\Psi_{ij}$, respectively. 

\end{lemma}
Note that, both gradient control terms  $\nabla_{p_i}\Phi_{ij}$ and $\nabla_{p_i}\Psi_{ij}$ can be expressed in a generic form as $\nabla_{p_i}\Phi_{ij}(p_i,p_j) = f(p_i,p_j)$ and $\nabla_{p_i}\Psi_{ij}(p_i,p_j) = g(p_i,p_j)$ respectively,
in which $g \in C^0$ is a continuous function.
We are now ready to introduce the control law that each robot~$i$ needs to run in order to keep its neighbors in the limited sensing zone $\trifov{i}$ near the desired point as the following
\begin{equation}\label{eq:u_i}
\begin{aligned}
\dot{s}_i &= - \underbrace{\sum \limits_{j \in \mathcal{N}_i^+} \nabla_{s_i} \Big ( \Phi_{ij}(s_i,s_j)+\Psi_{ij}(s_i,s_j) \Big )}_{u_i} \\
\end{aligned}
\end{equation}
%
%
%

\subsection{Asymmetric Interaction Graph Stability Analysis}

In this section we provide the main result of our work from \cite{santilli2019distributed} \cite{santilli:2020}, that is the stability analysis of the proposed topology control framework. From \cite{santilli2019distributed} we know that the overall Lyapunov function  for directed graphs $\Bar{V}: \mathbb{R}^{\dimp n} \times (-\pi,\pi]^n \rightarrow \mathbb{R}_+$ defined as
\begin{equation}\label{eq:V}
\begin{aligned}
    \Bar{V} (\mathbf{s}(t)) 
    &= \sum \limits_{i=1}^n \sum \limits_{j \in \mathcal{N}_i^+} \underbrace{\Big ( \Phi_{ij} \left (s_i,s_j \right ) + \Psi_{ij} \left (s_i, s_j \right) \Big )}_{\Bar{V}_{ij} \left (s_i,s_j \right )}
\end{aligned}
\end{equation}

To this end let us introduce the matrix $\overline{\mathcal{L}}$ defined as
\begin{equation}\label{eq:overlineL}
    \overline{\mathcal{L}} = \begin{bmatrix}
    \left ( \mathcal{B}^T \mathcal{B}_+ \right )\otimes I_{\dimp} & O_{ \dimp |\mathcal{E}| \times |\mathcal{E}|} \\
    O_{|\mathcal{E}| \times \dimp |\mathcal{E}|} & \left( \mathcal{B}_+^T \mathcal{B}_+ \right )
\end{bmatrix}
\end{equation}
where 
$I_{\dimp}$ is the $\dimp \times \dimp$ identity matrix and $O_{r \times r}$ is a $r \times r$ zeros matrix. We are now ready to state our main result.

\begin{theorem}(Theorem 1 in \cite{santilli2019distributed} and \cite{santilli:2020})\label{th:1}
Consider the multi-robot system $\dot{\mathbf{s}}(t) = \mathbf{u}(t)$ running control laws \eqref{eq:u_i}. Then, if the symmetric part of matrix $\overline{\mathcal{L}}$ in \eqref{eq:overlineL}, $ \overline{\mathcal{L}}^{\symm} = \frac{1}{2} \left( \overline{\mathcal{L}} + \overline{\mathcal{L}}^T \right )$, is positive semi-definite, the system is stable in the sense that if the energy $\Bar{V} (\mathbf{s}(t))$ is finite at time $t = t_0$ then it remains finite for all $t > t_0$.
\end{theorem}
In \cite{santilli:2020}, we show that
given that the Lyapunov function $\Bar{V}(s)$ is continuously differentiable we can define the level set~$\Delta_c = \{s : \Bar{V}(s) \leq c \}$ for any $c>0$ as a compact and invariant set with respect to the relative position of the robots. Arguments about the compactness of the level sets~$\Delta_c$ with respect with to the relative distances can be found in \cite{Dimarogonas:2008,Tanner:TAC:2007}. LaSalle's principle now guarantees that the system will converge to the largest invariant subset of $\{ s :  \dot{\Bar{V}}(s) = 0 \}$ by analyzing the nullspace of the matrix~$\overline{\mathcal{L}}^{\symm}$, $\left \{\xi \in \mathbb{R}^{3|\mathcal{E}|} : \overline{\mathcal{L}}^{\symm} \xi = 0 \right \}$.
From the complete proof of Theorem 1, we know that the Lyapunov derivative is defined as
\begin{equation}\label{eq:dotV:2}
\begin{aligned}
    \dot{\Bar{V}}(\mathbf{s}) 
     =& - \frac{1}{2} \xi^T \! \left [ \left( \overline{\mathcal{L}}\!  +\!  \overline{\mathcal{L}}^T \right ) \! + \! \left( \overline{\mathcal{L}}\!  -\!  \overline{\mathcal{L}}^T \right )  \right ] \xi =- \xi^T \overline{\mathcal{L}}^{\symm} \xi \leq 0 \\
\end{aligned}
\end{equation}
where $\overline{\mathcal{L}}$ is already defined in \eqref{eq:overlineL} and  
using the stacked vector of potential field gradients $\xi \in \mathbb{R}^{\dims|\mathcal{E}|} = \left [\xi_{xy}^T, \, \xi_{\theta}^T \right]^T$ where $\xi_{xy} \in \mathbb{R}^{2|\mathcal{E}|}$ and $\xi_{\theta} \in \mathbb{R}^{|\mathcal{E}|}$ are defined as
%
%
\begin{equation}\label{eq:xi}
\begin{aligned}
&\xi_{xy}=\!  \resizebox{0.92\hsize}{!}{%
    $\left [ \nabla_{x_{e_{1}(1)}} \Bar{V}_{e_1}^T, \nabla_{y_{e_{1}(1)}} \Bar{V}_{e_1}^T, \ldots, \nabla_{x_{e_{|\mathcal{E}|}(1)}} \Bar{V}_{e_{|\mathcal{E}|}}^T, \nabla_{y_{e_{|\mathcal{E}|}(1)}} \Bar{V}_{e_{|\mathcal{E}|}}^T \right ]^T$
    } \\
&\xi_{\theta} = \resizebox{.5\hsize}{!}{
$\left [ \nabla_{\theta_{e_{1}(1)}} \Bar{V}_{e_1}^T, \ldots, \nabla_{\theta_{e_{|\mathcal{E}|}(1)}} \Bar{V}_{e_{|\mathcal{E}|}}^T \right ]^T$}
\end{aligned}
\end{equation}
where $e_k(1)$ denotes the starting vertex $q_i$ of the $k$-th edge $(i,j)$, and thus $\nabla_{w_{e_{k}(1)}} \Bar{V}_{e_k} \in \mathbb{R}$ denotes the gradient with respect to the state variable $w_i \in \{x_i, y_i, \theta_i \}$ of potential function $\Bar{V}_{ij}$.
The Lyapunov time derivative  is in a typical quadratic form and its characteristics depend on the positive-semidefiniteness of symmetric block diagonal matrix $\overline{\mathcal{L}}^{\symm}$. 
\section{Distributed Adaptive FOV Control formulation}
So far we have derived a stable control law, as given in equation \eqref{eq:u_i}, to conduct a stable topology control of  MRS with limited FOV. 
In this section, we modify the control law to be adaptive for multi-robots directed topology control with limited FOVs. 
\subsection{Lyapunov Based Adaptive Control Design}\label{sub_sec:Lyp_adap}
For the sake of analysis, let us introduce a bijective indexing function $g(\cdot): \mathcal{V}\times \mathcal{V}\rightarrow \{1,...,|\mathcal{E}|\}$ to sort directed edges, that is $g(\mathcal{E})=[e_1,...,e_{|\mathcal{E}|}]$. Note that, since the graph is directed we have $e_h=g(e_{ij})\neq g(e_{ji}) $. Clearly, it is possible to define the inverse function such that if $e_h = g(e_{ij})$ then $e_{ij}=g^{-1}(e_h)$. Let us now consider an additional
stacked vector of gains $\bold{k}$ defined as
\begin{equation}\label{eq:k_vec}
\begin{aligned}
    \bold{k}=[k_{e_1},...,k_{e_{\mathcal{E}}}]
\end{aligned}    
\end{equation}
Then we have the following generalization of the distributed control law for a given robot $i$
\begin{equation}\label{eq:gen_con_law}
\begin{aligned}
    \dot{s}_i = \underbrace{- \sum_{j \in \mathcal{N}_i^+} k_{ij} \nabla_{s_i} \Bar{V}_{ij}(s_i,s_j)}_{\Bar{u}_i (s,\bold{k})}\\
    \dot{k}_{ij}= u_{ij}(s,\bold{k}), \quad \forall e_{ij}, j \in \mathcal{N}_i^+
\end{aligned}    
\end{equation}
As it will become clear later, two robots $i$ and $j$ which are the two endpoints of the same edge $e_{ij}$ will concur in the computation of the term $u_{ij}(s,\bold{k})$. 

Our goal is to design for each pair of robots $i$ and $j$ the control term $u_{ij}(s,\bold{k})$ to adaptively tune pairwise gains $k_{ij}$ to approximate the nominal pairwise interaction model given by $\Bar{V}_{ij}(s_i,s_j)$, independent of the network size. To design such an adaptive potential-based control law the following Lyapunov function can be considered 
\begin{equation}\label{eq:adp_V}
\begin{aligned}
    V(s,\bold{k})=\underbrace{\sum^n_{i=1}\sum_{j\in \mathcal{N}_i^+}k_{ij}\Bar{V}_{ij}(s_i,s_j)}_{\hat{V}(s,\bold{k})} + \underbrace{\sum^n_{i=1}\sum_{j\in \mathcal{N}_i^+} F_{ij}(p,\bold{k})}_{F(p,\bold{k})}
\end{aligned}    
\end{equation}
with $F_{ij}(p,\bold{k})$ the pairwise cost encoding the deviation of each pairwise interaction from its nominal model defined as
\begin{equation}\label{eq:adp_F}
\begin{aligned}
   F_{ij}(p,\bold{k})= \frac{1}{2} \| P_{ij}(\Bar{u}_i(p,\bold{k}))- m_{ij}(p,\bold{k})\|^2
\end{aligned}    
\end{equation}
with our desired pairwise model $m_{ij}$ defined as
\begin{equation}\label{eq:mij}
\begin{aligned}
   m_{ij}(p,\bold{k})= - \nabla_{p_i} \Bar{V}_{ij}(p_{ij})
\end{aligned}    
\end{equation}
that is, the anti-gradient of the potential function $\Bar{V}_{ij}$ with respect to robot $i$ and $P_{ij}(\Bar{u}_i(p,\bold{k}))$ the projection of the control input of robot $i$ over the line of sight of the robots $i$ and $j$ defined as 
\begin{equation}\label{eq:Pij}
\begin{aligned}
   P_{ij}(\Bar{u}_i(p,\bold{k}))=\left ( \frac{p_{ij}^T \Bar{u}_i(p,\bold{k})}{\|p_{ij}\|^2} \right ) p_{ij}
\end{aligned}    
\end{equation}
with $\Bar{u}_i(p,\bold{k})$ defined as in \eqref{eq:gen_con_law}. Moving forward, we remove the dependence from $(p,\bold{k})$ for the sake of readability. Note, we only solve for $k_{ij}$ explicitly using the state variables $p=[x, \quad y]^T$ because desired $\theta$ is solved for implicitly when optimal $k^*_{ij}$ is obtained as $x$ and $y$ coordinates are implicitly a function of polar coordinates.
With this idea, we can complete the control law given in \eqref{eq:gen_con_law} by introducing the following distributed control law for each edge-wise gain $k_{ij}$ with $(i,j)\in \mathcal{E}$
\begin{equation}\label{eq:uij}
\begin{aligned}
 u_{ij}(s,\bold{k})= - \nabla_{k_{ij}} F(p,\bold{k}) + w_{ij}
\end{aligned}    
\end{equation}
where the term $\nabla_{k_{ij}}F(p,\bold{k})$ is computed per edge in the directed graph as
\begin{equation}\label{eq:del_F}
\begin{aligned}
\nabla_{k_{ij}}F(p,\bold{k}) = \sum_{(i,h) \in\mathcal{E}} \nabla_{k_{ij}}F_{ih}(p,\bold{k}) 
\end{aligned}    
\end{equation}
Therefore, the overall $\nabla_{k}F(p,\bold{k})$ is given as the stacked vector of potential fields
$\nabla_{k}F(p,\bold{k})= [ \nabla_{k_{e_1}}F(p,\bold{k})^T,...,\nabla_{k_{e_{|\mathcal{E}|}}}F(p,\bold{k})^T]^T $. Another form can be shown as 
\begin{equation}\label{eq:del_F2}
\begin{aligned}
\nabla_{k}F(p,\bold{k})= [ \sum_{(i,h)\in \mathcal{E}} \nabla_{k_{e_{1}^{(i)}}}F_{ih}(p,\bold{k})^T,..., \\ 
\sum_{(i,h)\in \mathcal{E}}\nabla_{k_{e_{|\mathcal{E}|}^{(i)}}}F_{ih}(p,\bold{k})^T]^T 
\end{aligned}    
\end{equation}
  where $e_k^{(i)}$ is the $k^{th}$ edge $(i,h) \in \mathcal{E}$ with starting vertex(i.e. only outgoing directed edges) $i$. Note that the contribution to $\nabla_{k_{ij}}F(p,\bold{k})$ is from the directed edges only for all $(i,j)\in \mathcal{E}$ and the term $w_{ij}$.
As it will become clear in the next section, the form of
term $w_{ij}$ is the result of the Lyapunov-based design and is required to enforce the negative semi-definiteness of $\dot{V}(s,\bold{k})$.
\subsection{Adaptive Control Stability Analysis }
Let us now derive the overall extended dynamics of the MRS where each robot evolves according to \eqref{eq:gen_con_law} where gains are updated according to \eqref{eq:uij}. In particular, from $V(s,\bold{k})$ we obtain
\begin{equation}\label{eq:adp_law}
\begin{aligned}
\dot{s}&=-\nabla_{s+}\hat{V}(s,\bold{k}) \quad \dot{\bold{k}}&= -\nabla_{k}F(p,\bold{k}) + \bold{w}
\end{aligned}    
\end{equation}
where $s+$ in $\dot{s}$ indicates that it has contributions only from the starting vertex of
each edge\footnote{Refer to \cite{santilli2019distributed} for complete structure of $\dot{s}$. Also note that $\dot{p}$ can be represented in the same way as $\dot{s}$ as it consists of subset of state variables of $\dot{s}$.} and the stacked vector $\bold{w}$ is defined similarly to the stacked vector $\bold{k}$ as follows 
\begin{equation}\label{eq:w}
\begin{aligned}
\bold{w}=[w_{e_1},...,w_{e_{\mathcal{E}}}]
\end{aligned}    
\end{equation}
that is a vector collecting the terms $\{w_{ij}\}$ with $(i,j)\in \mathcal{E}$.
In order to demonstrate the theoretical properties of the adaptive control law given in \eqref{eq:adp_law}, the following convexity result of the cost $F_{ij}(p,\bold{k})$ is given.

\begin{theorem} \label{th:convx_F}
The potential function $F(p, \bold{k})$ given as a sum
of pairwise cost $F_{ij} (p, \bold{k})$ encoding the deviation of each pairwise interaction from its nominal model is convex in $\bold{k}$ for a fixed $p$.
\end{theorem}

\begin{proof}

First, we rewrite $P_{ij}(\Bar{u}_i(p,\bold{k}))$ as
\begin{equation}\label{eq:Pij2}
\begin{aligned}
  P_{ij}(\Bar{u}_i(p,\bold{k}))=\underbrace{\frac{1}{\|p_{ij}\|^2} (p_{ij} p_{ij}^T) }_{A_{ij}} \Bar{u}_i
  \end{aligned}    
\end{equation}
with $A_{ij}$ is a symmetric, positive semidefinite ``projection" matrix
with spectrum $\sigma(A_{ij})=\{0,1\}$. Therefore we can rewrite the function $F_{ij}(p,\bold{k})$ as
\begin{equation}\label{eq:new_Fij}
\begin{aligned}
  F_{ij}(p,\bold{k}) = \frac{1}{2}\|A_{ij}B_i \bold{k} - m_{ij}  \|^2
  \end{aligned}    
\end{equation}
where the matrix $B_i \in \mathbb{R}^{2\times |\mathcal{E}|}$ is defined as 
\begin{equation}\label{eq:Bi}
\begin{aligned}
B_i=[0,...,-\nabla_{p_i}\Bar{V}_{ij},0,...,-\nabla_{p_i}\Bar{V}_{ik},...,0]
  \end{aligned}    
\end{equation}
where all $-\nabla_{p_{i}}\Bar{V}_{ij}$ for which $j\in \mathcal{N}_i^+$ appear. 
At this point we compute the Hessian of $F_{ij}$ as defined in \eqref{eq:new_Fij} with respect to $\bold{k}$ for a fixed $p$ to obtain  
\begin{equation}\label{eq:Hessian}
\begin{aligned}
H_{\bold{k}}= (A_{ij}B_i)^T (A_{ij}B_i) = B_i^T A_{ij}^T A_{ij} B_i
  \end{aligned}    
\end{equation}
Since $A_{ij}$ is symmetric positive semidefinite, we can write $A_{ij}^2=Q^T\Lambda Q$ for a unitary matrix $Q$ and $\Lambda = \textbf{diag}\{0,1\}$ to obtain the form
\begin{equation}\label{eq:Hessian2}
\begin{aligned}
H_{\bold{k}}&=B_i^T A_{ij}^2 B_i = 
 (QB_i)^T \Lambda (QB_i)= Y^T \Lambda Y
  \end{aligned}    
\end{equation}
from which we notice that the construction of the hessian $H$ has spectrum $\sigma (H)= \{0,...,0,\sum_{i=1}^{|\mathcal{E}|}y^2_i\}$ where $Y=[y_i,...,y_{|\mathcal{E}|}]$, thus proving the convexity of $F_{ij}$.
To conclude, by invoking the sum-rule of convex functions \cite{boyd2004convex}, the result follows.
\end{proof}
We are now ready to state our main result on the stability and
convergence of the proposed adaptive potential-based control
law given in \eqref{eq:adp_law}.
\begin{theorem}\label{th:gain_stab}
Consider an MRS running \eqref{eq:adp_law} designed according to a generalized potential function as in \eqref{eq:adp_V} with $w_{ij}$
\begin{equation}\label{eq:wij2}
\begin{aligned}
w_{ij}= \frac{1}{\alpha_{ij}}(\gamma_{ij} + \frac{1}{|\mathcal{N}_i^+| + |\mathcal{N}_i^-|} \beta_{i} + \frac{1}{|\mathcal{N}_j^+| + |\mathcal{N}_j^-|} \beta_{j})
\end{aligned}    
\end{equation}
Then the MRS converges to a final state $[s_f^T,\bold{k}_f^T]$ reaching one of the critical points of $\hat{V}(s,\bold{k}_f)$ and the global optimum of $F(p_f,\bold{k})$.
\end{theorem}
\begin{proof}
Let us consider the Lyapunov function given in \eqref{eq:adp_V} and let us compute the time derivative as follows:
\begin{equation}\label{eq:adp_vdot}
\begin{aligned}
\dot{V}(s,\bold{k}) &= \nabla_s \hat{V}(s,\bold{k})^T \dot{s} + \nabla_{\bold{k}} \hat{V}(s,\bold{k})^T \dot{\bold{k}}\\
&+ \nabla_p F(p,\bold{k})^T \dot{p}+\nabla_{\bold{k}} F(p,\bold{k})^T \dot{\bold{k}}\\
&= - \xi_{\bold{k}}^T \, \overline{\mathcal{L}}^+ \, \xi_{\bold{k}} - \| \nabla_{\bold{k}}F(p,\bold{k})^T\|^2\\
&- \nabla_{p}F(p,\bold{k})^T\nabla_{p+}\hat{V}(p,\bold{k})\\
&-\nabla_{\bold{k}}\hat{V}(s,\bold{k})^T \nabla_{\bold{k}}F(p,\bold{k})\\
&+(\nabla_{\bold{k}}\hat{V}(s,\bold{k})^T+ \nabla_{\bold{k}}F(p,\bold{k})^T)
\bold{w}
\end{aligned}    
\end{equation}
Before we dwell into any further derivation, we will like to present the different structures of the elements of equation \eqref{eq:adp_vdot} for better clarity.  From equations $(25)$, $(29)$ in \cite{santilli2019distributed} and \eqref{eq:adp_law}, we already know the structures of $\nabla_s \hat{V}(s,\bold{k})$, $\dot{s}$, $\dot{p}$ or $\nabla_{p+}\hat{V}(p,\bold{k})$ and $\dot{\bold{k}}$, respectively. Equations \eqref{eq:del_F} and \eqref{eq:del_F2} provide the structure for $\nabla_{\bold{k}} F(p,\bold{k})$.  Below, we show the structures of $\nabla_p F(p,\bold{k})$ .
\begin{equation}\label{eq:gradF-derivation:1}
\begin{aligned}
    \nabla_{\mathbf{p}}F &=  
     \begin{aligned}  \Big [
    \sum \limits_{j \in \mathcal{N}_1^+} \nabla_{p_1} F_{1j}^T &+  \sum \limits_{j \in \mathcal{N}_1^-} \nabla_{p_1} F_{j1}^T \hspace{3mm}, \,\, \ldots \, \, ,  \\
    &+ \sum \limits_{j \in \mathcal{N}_n^+} \nabla_{p_n} F_{nj}^T + \sum \limits_{j \in \mathcal{N}_n^-} \nabla_{p_n} F_{jn}^T \Big ]^T
    \end{aligned}  \\
\end{aligned}
\end{equation}
With the gradient contributions related to the robot~$h$, we have
\begin{equation}\label{eq:gradF-derivation:2}
\begin{aligned}
    &\nabla_{p_h} F = \sum \limits_{j \in \mathcal{N}_h^+} \nabla_{p_h} F_{hj} + \sum \limits_{j \in \mathcal{N}_h^-} \nabla_{p_h} F_{jh} \\
    &= \begin{bmatrix}
    \sum \limits_{j \in \mathcal{N}_h^+} \nabla_{x_h} F_{hj} \\
    \sum \limits_{j \in \mathcal{N}_h^+} \nabla_{y_h} F_{hj}) \\
    \end{bmatrix} 
    -  \begin{bmatrix}
    \sum \limits_{j \in \mathcal{N}_h^-} \nabla_{x_j} F_{jh} \\
    \sum \limits_{j \in \mathcal{N}_h^-} \nabla_{y_j} F_{jh} \\
\end{bmatrix}
\end{aligned}
\end{equation}
Equation \eqref{eq:gradF-derivation:2} is valid because $\nabla_{\mathbf{p}}F$ too satisfies Lemma \ref{lemma:potential-property-psi} 
implying $\nabla_{\mathbf{p}}F$ is similar to $\nabla_{\mathbf{s}}\hat{V}$(from \cite{santilli2019distributed}) in structure. 
Finally, the structure of $\nabla_{\bold{k}} \hat{V}$ is given as
\begin{equation}\label{eq:delk_V2}
\begin{aligned}
\nabla_{k}\hat{V}(s,\bold{k})= 
 [  \nabla_{k_{e_{1}}}\hat{V}(s,\bold{k})^T,...,\nabla_{k_{e_{|\mathcal{E}|}}} \hat{V}(s,\bold{k})^T] \\ 
=[  \nabla_{k_{e_{1}^{(i)}}}\hat{V}_{ij}(s,\bold{k})^T,...,\nabla_{k_{e_{|\mathcal{E}|}^{(i)}}}\hat{V}_{ij}(s,\bold{k})^T]^T 
\end{aligned}    
\end{equation}
From the above derivation, it follows that in order to ensure negative-semidefiniteness of the Lyapunov derivative the following must hold:
\begin{equation}\label{eq:psd_cond}
\begin{aligned}
&- \nabla_{p}F(p,\bold{k})^T\nabla_{p+}\hat{V}(p,\bold{k})
-\nabla_{\bold{k}}\hat{V}(s,\bold{k})^T \nabla_{\bold{k}}F(p,\bold{k})\\
&+(\nabla_{\bold{k}}\hat{V}(s,\bold{k})^T+ \nabla_{\bold{k}}F(p,\bold{k})^T)
\bold{w}\leq 0
\end{aligned}    
\end{equation}
Therefore, to satisfy relation \eqref{eq:psd_cond} and consequently the negative semi-definiteness of $\dot{V}$, we derive a per-edge form of $\bold{w}$  as:
\begin{equation}\label{eq:psd_cond_edg}
\begin{aligned}
&\sum_{(i,j)\in \mathcal{E}}\left (\underbrace{\nabla_{k_{ij}}\hat{V}(s,\bold{k})^T+\nabla_{k_{ij}}F(p,\bold{k})^T}_{\alpha_{ij}}\right )w_{ij}   \leq   \\
&\sum_{(i,j)\in \mathcal{E}} \underbrace{\nabla_{k_{ij}} \hat{V}(s,\bold{k})^T \nabla_{k_{ij}} F(p,\bold{k})}_{\gamma_{ij}} +\\
& \sum_{(i,j)\in \mathcal{E}} ( \frac{1}{|\mathcal{N}_i^+| + |\mathcal{N}_i^-|} \underbrace{\nabla_{p_i} F(p,\bold{k})^T \nabla_{p+_i}\hat{V}(p,\bold{k})}_{\beta_i}+\\
& \frac{1}{|\mathcal{N}_j^+| + |\mathcal{N}_j^-|} \underbrace{\nabla_{p_j} F(p,\bold{k})^T \nabla_{p+_j}\hat{V}(p,\bold{k})}_{\beta_j})
\end{aligned}    
\end{equation}
where we can now rewrite \eqref{eq:psd_cond_edg} fully per-edge as:
\begin{equation}\label{eq:fin_wij}
\begin{aligned}
w_{ij}\leq \frac{1}{\alpha_{ij}}(\gamma_{ij} + \frac{1}{|\mathcal{N}_i^+| + |\mathcal{N}_i^-|} \beta_{i} + \frac{1}{|\mathcal{N}_j^+| + |\mathcal{N}_j^-|} \beta_{j})
\end{aligned}    
\end{equation}
At this point, by taking \eqref{eq:fin_wij} at the equality and plugging it into \eqref{eq:adp_vdot} , the following expression for the Lyapunov time derivative is obtained:
\begin{equation}\label{eq:fin_adp_V}
\begin{aligned}
\dot{V}(s,\bold{k})= - \xi_{\bold{k}}^T \, \overline{\mathcal{L}}^{\symm} \, \xi_{\bold{k}} - \| \nabla_{\bold{k}}F(p,\bold{k})\|^2 \leq 0
\end{aligned}    
\end{equation}
where $\xi_{\bold{k}}$ contains $\bold{k}$. 
From Theorem \ref{th:1} we know the first part of the above relation is negative semidefinite and from Theorem \ref{th:convx_F} that the function $F(p,\bold{k})$ is convex with respect to $\bold{k}$ for a fixed $s$. Therefore, by invoking the LaSalle's invariance theorem the result follows.
\end{proof}
\begin{remark}
Note that a singularity situation with the control term $w_{ij}$ will never occur  for $\alpha_{ij}=0$ because $\dot{V}(s,\bold{k})$ is negative semi-definite and $V(s,\bold{k})$ is positive-definite as per Theorem \ref{th:gain_stab}.
\end{remark}

\section{Resilient Control Design}
In this section we derive an $H_{\infty}$ control protocol with the intention of making the potential based control framework with adaptive gains for MRS resilient to external attacks that induce faults in the sensor and actuator which can render the MRS unstable. 
We know from the original form of the Gaussian potential function $\Psi_{ij}$ given in \cite{santilli2019distributed}, which represents the quality of interaction between robots, that the variances, $\sigma_x$ and $\sigma_y$,  are influenced by the sensor measurements in $p=[x,y]^T$. Consequently, since the MRS control framework is based on potential based control, a typical part of the overall control law is given by $\nabla_{p_i}\Psi_{ij}(p_i,p_j)$ as also shown in equation \eqref{eq:u_i}. Therefore it is reasonable to assume faults in the state variables $x$ and $y$ only. In this regard, we take the approach provided in \cite{chen2019resilient} that is to simultaneously tackle sensor and actuator fault for which we design an $H_{\infty}$ control protocol specifically for the sensor fault but also model the actuator fault as a part of the system disturbances.   In this direction, we define fault equations as
\begin{equation}\label{eq:sens_fault2}
\begin{aligned}
\Bar{p}_i = p_i + \delta^p_i \quad \hat{u}_{p_i} = \Bar{u}_{p_i} + \delta^u_i
\end{aligned}    
\end{equation}
where $\Bar{p}_i$ is the faulty sensor measurement, $p_i$ is the actual measurement which is unknown, $\delta^p_i$ is the unknown sensor fault, $\hat{u}_{p_i}$ is the faulty control input only with respect to states $p$, $\Bar{u}_{p_i}$ is the desired control input again only with respect to states $p$ such that the overall control action is $\Bar{u}_i=[\Bar{u}_{p_i}^T,\Bar{u}_{\theta_i} ]^T$ and $\delta^u_i$ is the unknown actuator fault induced by the attacker. Note that, it is trivial to repeat this analysis with actuator fault with respect to $\theta$ but for consistency of dimensions, we disregard it for this analysis. Now, to proceed with designing the $H_{\infty}$ control protocol, we first define an error dynamics as follows
\begin{equation}\label{eq:er}
\begin{aligned}
e_i = \Bar{p}_i - \hat{p}_i - \hat{\delta}^p_i
\end{aligned}    
\end{equation}
where $\hat{p}_i$ is an estimate of the uncorrupted state variables $p_i$ and $\hat{\delta}^p_i$ is an estimate of the unknown sensor fault $\delta^p_i$. 
Thus, the $H_{\infty}$ control protocols are given by
\begin{equation}\label{eq:new_dyn}
\begin{aligned}
\dot{\hat{p}}_i = \Bar{u}_{p_i} + w_i \quad w_i =(F_1 + F_2)e_i \quad \dot{\hat{\delta}}^p_i = -F_1 e_i \\
\end{aligned}    
\end{equation}
where  $F_1$ and $F_2$ are the observer gains. 
Before, we provide the results from a theorem which will prove the boundedness of the error, we will like to provide the following definition and lemma on general static output-feedback control design . 
\begin{definition}\label{def:hinf}
Define a linear time-invariant system $\dot{p}=\Bar{A}p+\Bar{B}u+\Bar{D}d$,$y=\Bar{C}p$, where $u$,$y$ and $d$ denote the system input, output, and disturbance, respectively. Define a performance output $w$ as $\|w\|^2=p^T \Bar{Q}p+ u^T \Bar{R}u$ for $\Bar{Q}\geq 0$ and $\Bar{R}>0$. The system $L_2$ gain is said to be bounded or attenuated by $\gamma$ if the $L_2$ norms of $w$ and $d$ satisfy: $\frac{\int_{0}^{\infty} \|w\|^2 dt}{\int_{0}^{\infty} \|d\|^2 dt} = \frac{\int_{0}^{\infty} (p^T \Bar{Q}p + u^T \Bar{R}u) dt}{\int_{0}^{\infty} (d^Td) dt}\leq \gamma^2$
\end{definition}

\begin{lemma}(Theorem 1 in \cite{gadewadikar2006necessary})\label{lem:stab}
Assume that $(\Bar{A},\sqrt{\Bar{Q}})$ is detectable with $\Bar{Q\geq 0}$. Then, the system considered in definition \ref{def:hinf} is output-feedback stabilizable with $L_2$ gain bounded by $\gamma$, if and only if i) there exist matrices $\Bar{K}$, $M$, and $\Bar{P}$ such that 
\begin{equation}\label{eq:hinfcond}
\begin{aligned}
&\Bar{K}\Bar{C}=\Bar{R}^{-1}(\Bar{B}^T\Bar{P}+M)\\
&\Bar{P}\Bar{A}+\Bar{A}^T\Bar{P}+\Bar{Q}+\gamma^{-2}\Bar{P}\Bar{D}\Bar{D}^T\Bar{P}+ M^T \Bar{R}^{-1}M\\
&=\Bar{P}\Bar{B}\Bar{R}^{-1}\Bar{B}^T
\Bar{P}
\end{aligned}    
\end{equation}
and ii) $(\Bar{A},\Bar{B})$ is stabilizable and $(\Bar{A},\Bar{C})$ is detectable.
\end{lemma}
Before we provide the final result we define the estimation error for the sensor fault as $\Tilde{\delta}^p_i=\delta^p_i-\hat{\delta}^p_i$.
\begin{assumption}\label{asum:1}
The directed graph $\mathcal{G}$ contains a spanning tree with
the leader as its root.
\end{assumption}
\begin{assumption}\label{asum:2}
$\delta^p_i$ is unbounded but the derivative of $\dot{\delta}^p_i$ of the sensor fault in \eqref{eq:sens_fault2} is bounded. The actuator fault $\delta_i^u$ is bounded.
\end{assumption}
\begin{theorem}(Modified version of Theorem 4 in \cite{chen2019resilient} )
Suppose that the graph $\mathcal{G}$ satisfies Assumption \ref{asum:1}, the sensor and actuator faults satisfy Assumption \ref{asum:2}. Design $F_1$ and $F_2$ such that $\Bar{K}=[F_2^T,-F_1^T]$ follows \eqref{eq:hinfcond} in Lemma \ref{lem:stab}. Then for the MRS topology control with limited FOV under the sensor and actuator faults \eqref{eq:sens_fault2} is solved by the $H_{\infty}$ control protocol \eqref{eq:new_dyn}. Moreover, the $L_2$ gains of the errors $e_i$ and $\Tilde{\delta}^p_i$ are bounded in terms of the $L_2$ norms of disturbance $d_i$.
\begin{proof}

Define $d=[d_i^T, d_2^T,...,d_n^T]^T$, $e=[e_1^T,e_2^T,...,e_n^T]^T$, and $\Tilde{\delta}^p=[\Tilde{\delta}^p_1,\Tilde{\delta}^p_2,...,\Tilde{\delta}^p_n]^T$. From $\dot{\mathbf{s}}(t) = \mathbf{u}(t)$ and \eqref{eq:new_dyn}, differentiating $e_i$ in \eqref{eq:er}
with respect to time $t$ yields 
\begin{equation}\label{eq:err_dyn}
\begin{aligned}
&\dot{e}_i= \dot{\Bar{p}}_i - \dot{\hat{p}}_i - \dot{\hat{\delta}}^p\\
&=\hat{u}_{p_i} -  ( \Bar{u}_{p_i} + w_i) + F_1 e_i = - F_2 e_i + \delta_i^u \\
\end{aligned}    
\end{equation}
Thus, using \eqref{eq:new_dyn} we can show
\begin{equation}\label{eq:hinfcond23}
\begin{aligned}
 \begin{bmatrix} \dot{e}_i & \dot{\Tilde{\delta}}^p_i  \end{bmatrix}^T=A_{F_1}
 \begin{bmatrix} e_i & \Tilde{\delta}^p_i  \end{bmatrix}^T + d_i
\end{aligned}    
\end{equation}
where $A_{F_1}=[a_{f_{ij}}^1]$ with $a_{f_{11}}^1= -F_2$,$a_{f_{12}}^1= 0$,$a_{f_{21}}^1= F_1$,$a_{f_{22}}^1= 0$ and $d_i=[\delta^u_i ,\quad \dot{\delta}^p_i]^T$. In the following, we will show the stabilization of \eqref{eq:hinfcond23} can be achieved by appropriately designing $F_1$ and $F_2$. To do this, we transform \eqref{eq:hinfcond23} to the following full state output feedback control system 
\begin{equation}\label{eq:fdbck_sys}
\begin{aligned}
\dot{x}_T \triangleq \Bar{A}x_T + \Bar{B}u_T + \Bar{D}d_i, \quad y_T \triangleq \Bar{C}x_T
\end{aligned}    
\end{equation}
where $\Bar{A}=[A,0;0,0]$ with $A=I_n$, $\Bar{B}=\Bar{D}=[I_n,0;0,I_n]$ and $\Bar{C}=[I_n,0]$. Moreover, define the controller $u_T= -\Bar{K y_T}$, where $\Bar{K}=[K_1^T,K_2^T]^T$. Therefore, choosing $k_1=F_2 + I_n$ and $K_2=-F_1$ makes \eqref{eq:fdbck_sys} equivalent to \eqref{eq:hinfcond23}. At this point after showing \eqref{eq:fdbck_sys} is equivalent to \eqref{eq:hinfcond23}, for the sake of brevity, the rest of the proof is not shown as it is exactly the same in Theorem 4 in \cite{chen2019resilient}. In brief, it is shown that by using  Definition \ref{def:hinf} and Lemma \ref{lem:stab} and as a result of the equivalence relation, the $[e^T, \tilde{\delta}^{p^T}]^T$ is bounded.


To finish the proof , we will show the MRS control is stable under the proposed control. We define a new Lyapunov function $V_e=\frac{1}{2}\|e\|^2$ with respect to the bounded error $e$ as proved above. Taking the derivative we have $\dot{V_e}=-e^Te$.
Therefore, 
the complete system in \eqref{eq:adp_V} becomes  $\dot{V}=\dot{\hat{V}}+\dot{F}+\dot{V}_e \leq 0$, which implies that under the $H_{\infty}$ control protocols all robots satisfy the individual topological objectives represented by $\hat{V}$ and $F$, thus maintaining a stable MRS in the presence of faults. 
\end{proof}

\end{theorem}
\section{Decentralized Reinforcement learning to solve Adaptive controller online}
In this section we provide an alternative method to solve for the optimal function $F(p,\bold{k})$, that is solving for optimal pairwise gains $k^*_{ij}$ using an online Q-Learning formulation.
\subsection{Direct Estimation of Q-Function}
First, we will convert the continuous dynamics of the system to discrete form. The discrete dynamics per-robot is given by
\begin{equation}\label{eq:q-gen}
\begin{aligned}
s_i(t+1)= s_i(t) + \Delta_t\Bar{u}_i(t)
\end{aligned}    
\end{equation}
where $\Delta_t$ is the unit time step of the discrete time system.
Now, we know the general form of the Q-function is given as
\begin{equation}\label{eq:q-gen2}
\begin{aligned}
Q(x_t,u_t)= c(x_t,u_t)+ \gamma Q(f(x_{t+1},u_{t+1}))
\end{aligned}    
\end{equation}
where $c(x_t,u_t)$ cost per stage and $\gamma$ is the discount factor.
For a direct estimation of the above Q-function , we know intuitively that \eqref{eq:new_Fij}, which resembles a typical \emph{least squares} form, can be represented to  look like
\begin{equation}\label{eq:q-est1}
\begin{aligned}
c(x_t,u_t)= \phi_t^T \Theta_q
\end{aligned}    
\end{equation}
where $q$ denotes policy change steps. The form in \eqref{eq:q-est1} and the general Q-Learning algorithm is well studied in literature \cite{lewis2012reinforcement}.
Therefore, we define a Q-function in the same way as
\begin{equation}\label{eq:q-est2}
\begin{aligned}
\small \underbrace{(-A_{ij}(t)B_{ij}(t) - \gamma^l (-A_{ij}(t+1)B_{ij}(t+1)) )^T}_{\phi_t^T}\underbrace{\bold{k}_{ij}}_{\Theta} = \underbrace{m_{ij}(t)}_{c(x_t,u_t)}
\end{aligned}    
\end{equation}
The above is solved by each robot $i$ for each of its out-going edges $j\in \mathcal{N}_i^+$ using recursive least square (RLS) method. The RLS method will help solve for the exact estimate $\hat{\Theta}$ of $\Theta$. However, as already well know, for the RLS to converge, $\phi_t$ must satify the persistence of excitation (PE) condition given by the following definition.
\begin{definition}\label{def:pe}
The matrix sequence $\phi_t$ is said to be persistently exciting (PE) if for some constant $T_{in}$ there exists positive constants $\epsilon_0$ and $\epsilon_1$ such that
\begin{equation}\label{eq:pe2}
\begin{aligned}
\epsilon_0I\leq \frac{1}{T_{in}}\sum_{l=1}^{T_{in}} \phi_{t-l}\phi_{t-l}^T\leq \epsilon_1I \forall t\geq t_0>0
\end{aligned}    
\end{equation}
where $t$ is the total time step, $l$ is the time step since the last policy change and $T_{in}$ is the time interval between consecutive $q$ policy changes. 
\end{definition}
\subsection{Convergence to Optimal Policy }
In this section, we show that the our Q-Learning algorithm formulated using the general form in \cite{lewis2012reinforcement} converges to solve for optimal gain $k^*_{ij}$. Therefore, before we present the convergence results in Theorem \ref{th:q_conv}, we provide Lemma \ref{lem:rls_conv}.
\begin{lemma}(Theorem 2 in \cite{islam2019recursive})\label{lem:rls_conv}
Suppose $\phi_t$ is PE. Then, for $l=1$ to $T_{in}$ the RLS estimate $\hat{\Theta}_m$ converges to actual $\Theta$ which also minimizes 
\begin{equation}\label{eq:rls_min}
\begin{aligned}
\min_{\Theta} \| c_t - \phi_t^T \Theta      \|^2
\end{aligned}    
\end{equation}
\end{lemma}
\begin{theorem}\label{th:q_conv}
Suppose $\phi_t$ is persistently excited satisfying \eqref{eq:pe2} and consequently satisfies Lemma \ref{lem:rls_conv}, then there $\exists T_{in}<\infty$ such that the adaptive policy iteration mechanism converges to the optimal policy with $k^*_{ij}$

\end{theorem}
\begin{proof}
Given PE (Definition \ref{def:pe}) exists and invoking Lemma \ref{lem:rls_conv}, we know that $\hat{\Theta}\rightarrow \Theta$ which also implies \eqref{eq:rls_min} is minimized.
Thus, a discount factor $\gamma \in [0,1]$ where $\gamma^l\rightarrow 0$ as $l\rightarrow \infty$ implies \eqref{eq:q-est2} looks like
\begin{equation}\label{eq:q-est23}
\begin{aligned}
(-A_{ij}(t)B_{ij}(t) )^T\bold{k}_{ij} = m_{ij}(t)
\end{aligned}    
\end{equation}
Now, for better understanding,
consider an robot $i$ having a control action
\begin{equation}\label{eq:ui1}
\begin{aligned}
\Bar{u}_i = \sum_{j \in \mathcal{N}_i^+} k_{ij} v = -v \left ( \sum_{j \in \mathcal{N}_i^+} k_{ij}  \right )
\end{aligned}    
\end{equation}
Recall, that \eqref{eq:q-est23} can be represented as an
optimization objective $O(k)$ for robot $i$ as
\begin{equation}\label{eq:obj11}
\begin{aligned}
O_i (\bold{k}) = \sum_{j \in \mathcal{N}_i^+}\|P_{ij}(\Bar{u}_i)- m_{ij}\|^2  = \sum_{j \in \mathcal{N}_i^+}\|P_{ij}(\Bar{u}_i)+ v\|^2
\end{aligned}    
\end{equation}
Then assume a common model of $v= \xi x_{ij}$ for $\xi \in \mathbb{R}$, 
the robot $i$s contribution becomes

\begin{equation}\label{eq:obj23}
\begin{aligned}
O_i({\bold{k}}) = \sum_{j\in \mathcal{N}_i^+} \xi^2 \left (1 - \left ( \sum_{j\in \mathcal{N}_i^+} k_{ij}     \right ) \right )^2 |\mathcal{N}^+_i|\|q\|^2
\end{aligned}    
\end{equation}
taking $x_{ij}=q$ for all $j\in \mathcal{N}^+_i$.
Thus, it is apparent that the optimal gains for robot $i$ are
$k^*_{ij} = 1/|\mathcal{N}^+_i|, \forall e_{ij} , j \in \mathcal{N}_i^+$
where the global minimum
$O_i(k) = 0$ is achieved, which completes the proof. 
We can claim that
our adaptive behavior is independent of the network size and obtains optimal $k^*_{ij}$ using \emph{policy iteration} based Q-learning. 
\end{proof}

\section{Simulation Results}\label{sec:simulations}
In this section, we discuss results from
 a MATLAB simulation over $350$ seconds with $6$ robots,with robot $6$ as the leader with exogenous input, in a directed spanning tree configuration for a leader-follower scenario\footnote{Refer to the submitted video for further details}. 
\subsection{Pairwise adaptive gain for MRS}
\def\figsize{0.89}
\def\figsizee{0.95}
\def\figsizeee{0.92}
\def\subfigsize{0.5}
\def\spacereduction{-0.15cm}

\begin{figure}
\centering
\begin{subfigure}[b]{0.47\textwidth}
   \includegraphics[width=1\linewidth]{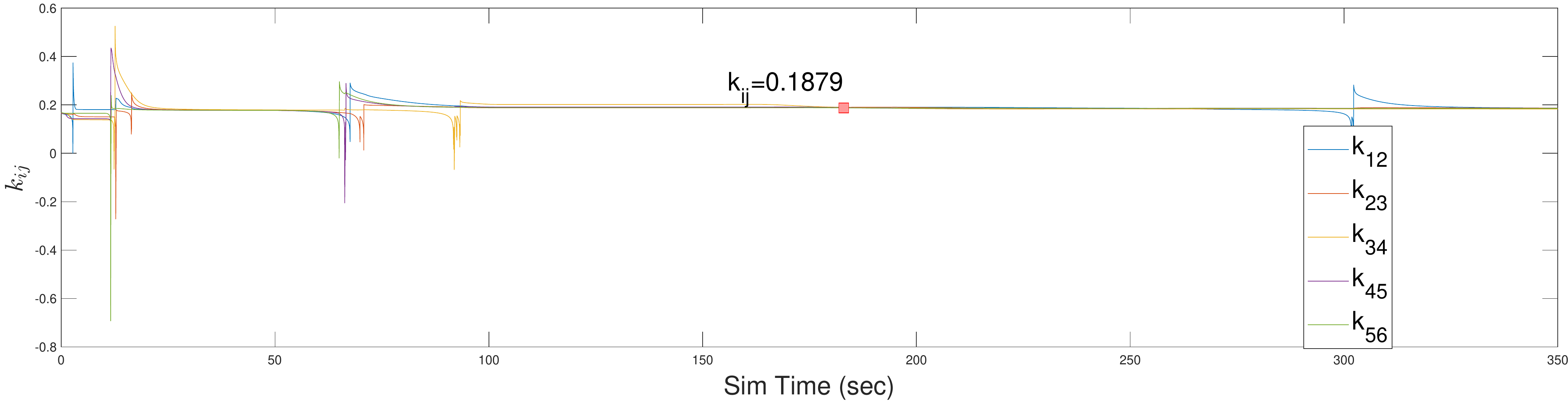}
   \caption{}
   \label{fig:kij} 
\end{subfigure}

\begin{subfigure}[b]{0.47\textwidth}
   \includegraphics[width=1\linewidth]{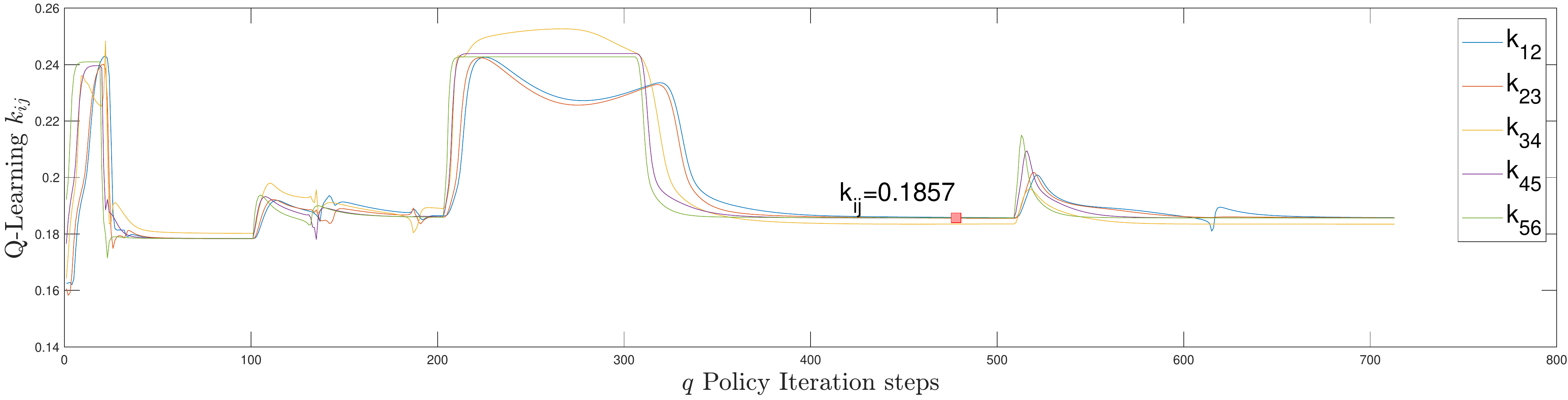}
   \caption{}
   \label{fig:qkij}
\end{subfigure}

\caption[Two numerical solutions]{(a) $k_{ij}$ computed using adaptive law equation \eqref{eq:adp_law}. (b)$k_{ij}$ computed using policy iteration based Q-Learning.}
\end{figure}

\begin{figure}[t!]
    \centering
    \includegraphics[width=1\columnwidth]{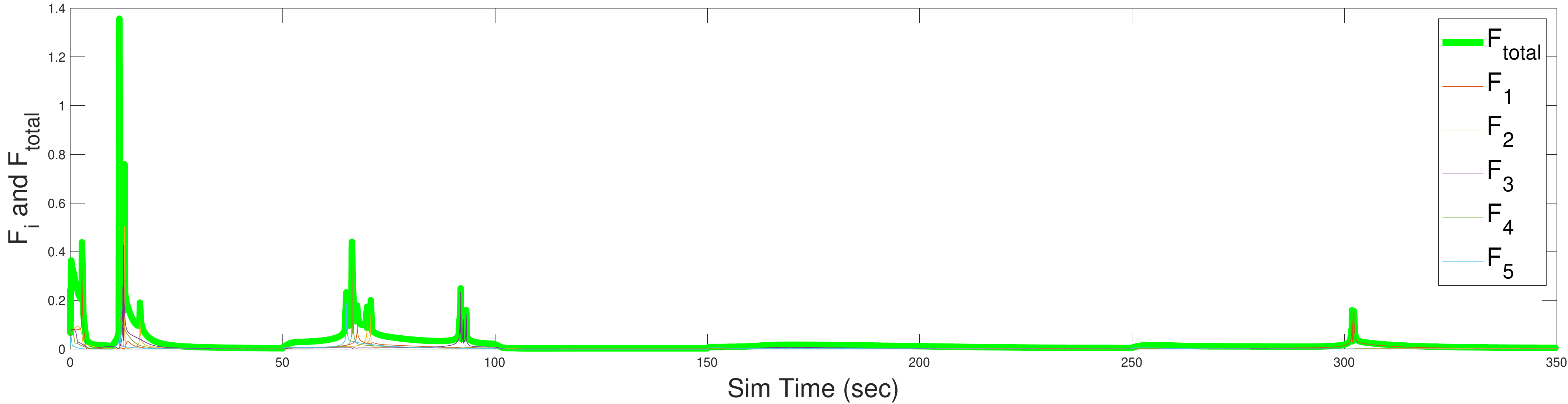}
    \caption{Plot of $F_i(p,k)$ and total $F(p,k)$, the pairwise cost function encoding the deviation of pairwise interactions from its nominal model. }\label{fig:F}
\end{figure}
Figure \ref{fig:kij} shows the computed $k_{ij}$ for each directed edge for over $350s$ of the total simulation time computed using equation \eqref{eq:adp_law}. Corresponding to the the computed $k_{ij}$, Figure \ref{fig:F} shows the plot of the total cost value of the objective function $F$ and per robot cost $F_i$. Clearly, the optimum value is achieved over time as after every peak, which represents a change in the actuation of the leader robot $6$, the $F_i$ value converges to zero.
\subsection{Resiliency from induced sensor and actuator fault on selected robots of MRS}
\def\figsize{0.89}
\def\figsizee{0.95}
\def\figsizeee{0.92}
\def\subfigsize{0.5}
\def\spacereduction{-0.15cm}

\begin{figure}
\centering
\begin{subfigure}[b]{0.47\textwidth}
   \includegraphics[width=1\linewidth]{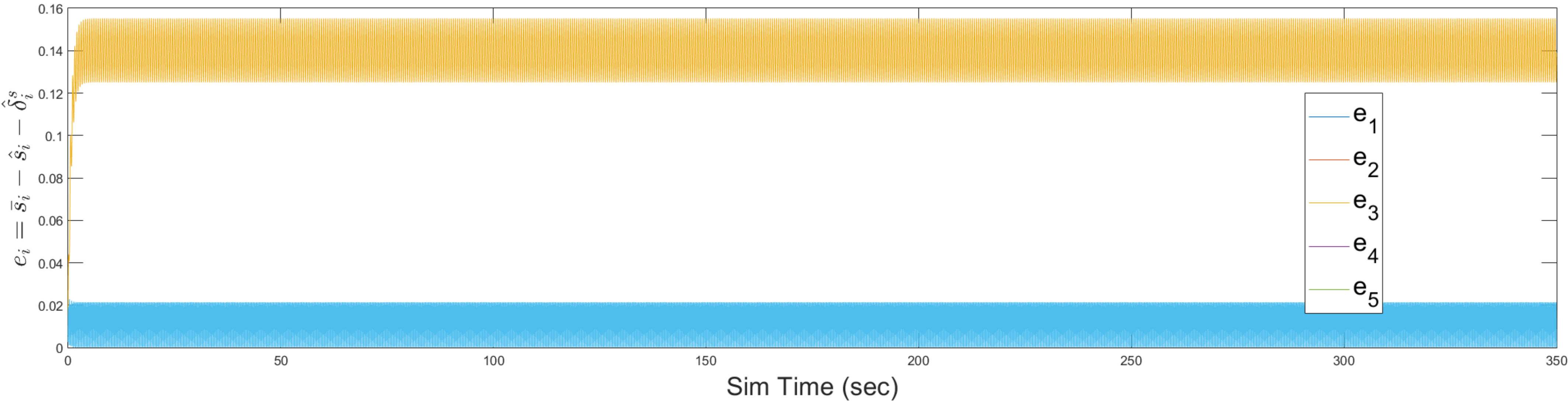}
   \caption{}
   \label{fig:er} 
\end{subfigure}

\begin{subfigure}[b]{0.47\textwidth}
   \includegraphics[width=1\linewidth]{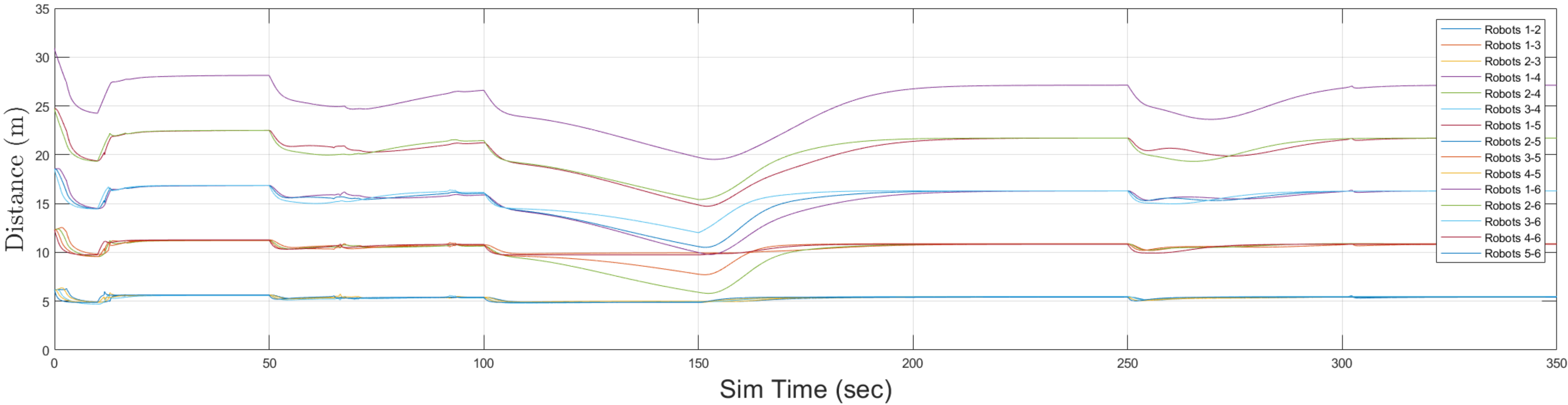}
   \caption{}
   \label{fig:robot_dist}
\end{subfigure}

\caption[Two numerical solutions]{(a) Plot of error dynamics $e_i$ for each robot as given by \eqref{eq:er}. (b) Plot of distance of each robot from the other robots in MRS for complete simulation time of $350s$.}
\end{figure}
For this simulation scenario, robot $3$ is induced with an unbounded sensor fault $0.2t$ 
and all robots with bounded actuation fault $1.5sin(2\pi t)$ where both comply with Assumption \ref{asum:2}. The yellow plot in Figure \ref{fig:er} shows the error dynamics of robot $3$, where clearly the error induced due to the sensors fault in Robot $3$ is attenuated. Similarly the oscillatory, bounded actuator fault is visible in the oscillation of the error dynamics of all the robots. Now, Figure \ref{fig:robot_dist} shows the uniform distances measured between the robots with no indication of the error induced by the unbounded sensor fault or the bounded actuator fault indicating that the global topology control objective was attained despite of induced faults.

\subsection{Q-Learning based pairwise gain determination}
\def\figsize{0.89}
\def\figsizee{0.95}
\def\figsizeee{0.92}
\def\subfigsize{0.5}
\def\spacereduction{-0.15cm}

\begin{figure}
\centering
\begin{subfigure}[b]{0.47\textwidth}
   \includegraphics[width=1\linewidth]{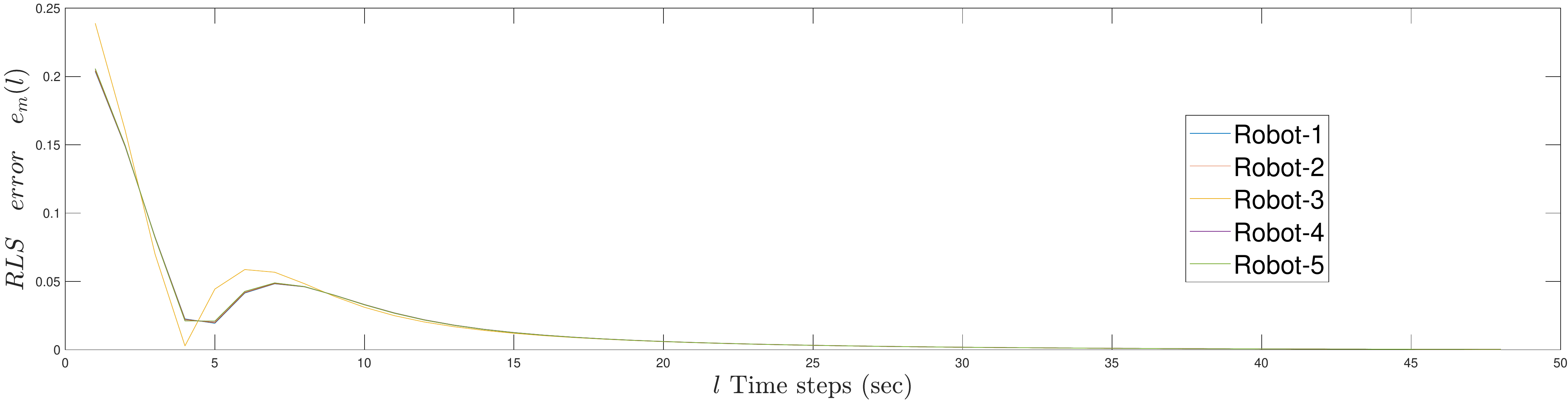}
   \caption{}
   \label{fig:rlser} 
\end{subfigure}

\begin{subfigure}[b]{0.47\textwidth}
   \includegraphics[width=1\linewidth]{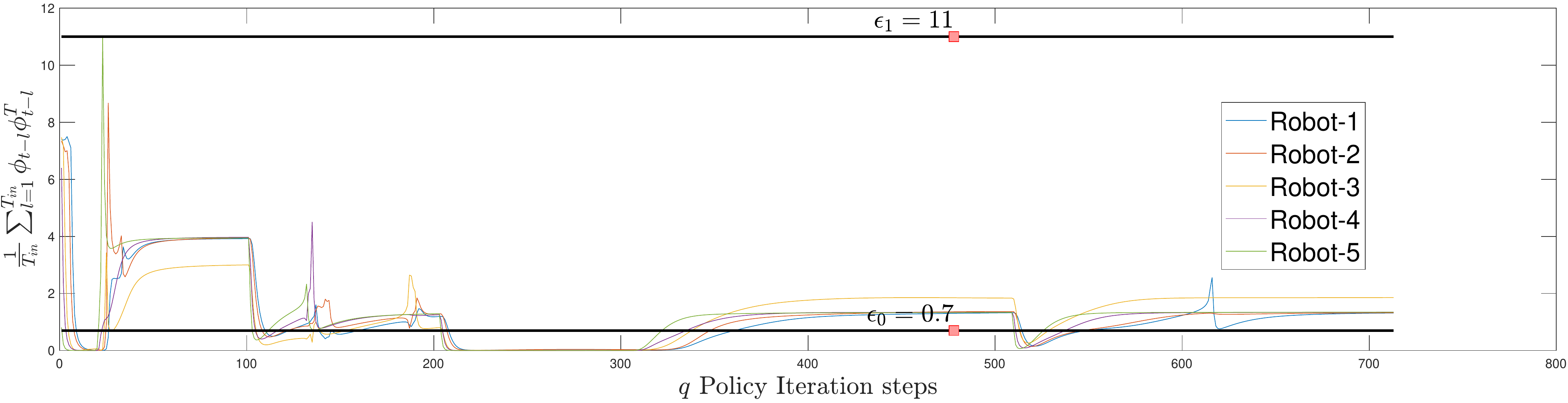}
   \caption{}
   \label{fig:pe}
\end{subfigure}

\caption[Two numerical solutions]{(a) Plot of RLS error for each robot over $l$ time-steps between every policy change. (b) Plot of persistence of excitation (PE) condition for each robot given by relation \eqref{eq:pe2}}
\end{figure}
Note that $k_{ij}$ solved using adaptive laws \eqref{eq:adp_law} and $k_{ij}$ solved using Q-Learning are very similar as denoted by the data points on the plots of Figure \ref{fig:kij} and \ref{fig:qkij}. Further, Figure~\ref{fig:rlser} shows the RLS error convergence over $l=50$ iterations between every policy change and Figure \ref{fig:pe} shows that the PE condition \eqref{eq:pe2} is satisfied over $q$ policy change steps for $\epsilon_0 =0.7$ and $\epsilon_1 =11$. The implication of Figure \ref{fig:pe} is that the PE condition is satisfied when there is an exogenous input provided to the leader robot $6$. 

\section{Conclusion}\label{sec:conclusion}
In this paper, we derived an adaptive control law capable of maintaining the interactions among robots in MRS such that the pairwise ``prescribed'' performance is achieved which is then coupled with $H_{\infty}$ control protocols to make the MRS resilient to additive sensor and actuator faults. Finally, we show a policy iteration based Q-Learning implementation for the discrete-time form of the MRS to solve for optimal gain $k^*_{ij}$. 
Numerical simulations are provided to support the theory.
\bibliographystyle{IEEEtran}
\bibliography{ref}

\end{document}